\documentclass[12pt,onecolumn,draft]{IEEEtran} 

\usepackage{amsmath,amssymb,amsfonts}
\usepackage[final]{graphicx}

\newtheorem{corollary}{Corollary}

\newtheorem{lemma}{Lemma}

\def \M {{\mathcal M }}

\def \L {{\mathcal{L}}}

\def \figtab{14cm}
\def \figspace{\vspace*{-5mm}}

\begin{document}
\title{\Large Outage Probability in \mbox{$\eta$-$\mu$/\mbox{$\eta$-$\mu$} Interference-limited Scenarios}
\thanks{*J. F. Paris is with the Dept. Ingenier\'ia de Comunicaciones, Universidad de M\'alaga, M\'alaga, Spain.
This work is partially supported by the Spanish Government under
project TEC2011-25473 and the European Program FEDER.}
\thanks{**This work has been submitted to the IEEE for possible
publication. Copyright may be transferred without notice, after
which this version may no longer be accessible.}}

\author{
\vspace{-3mm}
\authorblockN{\normalsize Jos\'e F. Paris}
\vspace{-15mm}}


\maketitle

\begin{abstract}
In this paper exact closed-form expressions are derived for the
outage probability (OP) in scenarios where both the signal of
interest (SOI) and the interfering signals experience $\eta$-$\mu$
fading and the background noise can be neglected. With the only
assumption that the $\mu$ parameter is a positive integer number
for the interfering signals, the derived expressions are given in
elementary terms for maximal ratio combining (MRC) with
independent branches. The analysis is also valid when the $\mu$
parameters of the pre-combining SOI power envelopes are positive
integer or half-integer numbers and the SOI is formed at the
receiver from spatially \mbox{correlated MRC.}
\end{abstract}

\vspace{-1mm}

\vspace{-4mm}

\IEEEpeerreviewmaketitle

\section{Introduction}

\PARstart{T}{he} $\eta$-$\mu$ fading distribution has been
proposed to model a general non-line-of-sight (NLOS) propagation
scenario. By two shape parameters $\eta$ and $\mu$, this model
includes some classical fading distributions as particular cases,
e.g. Nakagami-$q$ (Hoyt), one-sided Gaussian, Rayleigh and
Nakagami-$m$. The fitting of the $\eta$-$\mu$ distribution to
experimental data is better than that achieved by the classical
distributions previously mentioned. A detailed description of the
$\eta$-$\mu$ fading model can be found \mbox{in \cite{Yacoub07}.}

This paper focuses on outage probability (OP) analysis for
wireless communications systems where both the signal of interest
(SOI) and the cochannel interference (CCI) signals experience
$\eta$-$\mu$ fading and the background noise can be neglected. The
OP analysis for $\eta$-$\mu$ fading channels in which only the
background noise is present was recently published in
\cite{Paris10}. \mbox{A detailed} account of OP analysis for
interference-limited systems can be found in \mbox{\cite[chapter
10]{Simon05}} and references therein. Specifically, in
\mbox{\cite[eq. 10.17]{Simon05}} a closed-form expression is
derived for the OP in the \mbox{Nakagami-$m$/Nakagami-$m$}
interference-limited scenario assuming independent and identically
distributed (i.i.d) receive maximal ratio combining (MRC) for the
SOI, i.i.d interfering signals, and certain restrictions for the
values of the Nakagami-$m$ parameters of both the SOI and the
interferers. In fact, this expression for the
\mbox{Nakagami-$m$/Nakagami-$m$} scenario includes, as particular
cases, several classical results derived in literature
\mbox{\cite{Sowerby88}-\cite{Tellambura95}}. More recently, a
further generalization for the \mbox{Nakagami-$m$/Nakagami-$m$}
interference-limited scenario was presented \mbox{in
\cite{Romero07},} where a closed-form expression was derived in
\cite[eq. 18]{Romero07} for i.i.d MRC with the only restriction
that the product of the $m$ parameter and the number of MRC
branches for the SOI is a positive integer number. However, to the
best of the author's knowledge, no closed-form results in
elementary terms for the \mbox{$\eta$-$\mu$/$\eta$-$\mu$}
interference-limited scenario, which consequently generalize
current results for the \mbox{Nakagami-$m$/Nakagami-$m$} scenario,
are found in literature.

In this paper, we derive exact closed-form expressions for the OP
of the \mbox{$\eta$-$\mu$/$\eta$-$\mu$} interference-limited
scenario in elementary terms, with the only assumption that the
$\mu$ parameter is a positive integer number for the interfering
signals. Such analysis is further extended to scenarios in which
the SOI is formed from spatially correlated MRC. In connection
with the \mbox{Nakagami-$m$/Nakagami-$m$} scenario, the derived
expressions complement \cite[eq. 10.17]{Simon05} and \cite[eq.
18]{Romero07} as long as they are valid if real values of $m$ are
assumed for the pre-combining SOI power envelopes or if the SOI is
formed from spatially correlated MRC. In addition, the presented
OP analysis includes other interesting scenarios, e.g. the
$\eta$-$\mu$/Rayleigh interference-limited scenario with no
assumptions on the SOI and the interferers parameters.

The remainder of this paper is organized as follows.
In \mbox{Section \ref{stat}} the key statistical results for $\eta$-$\mu$/$\eta$-$\mu$
random variables (RVs) are derived.
The OP analysis
is discussed in \mbox{Section
\ref{analysis}} and the numerical results in \mbox{Section \ref{resultados}}.
Finally, some conclusions are given in
\mbox{Section \ref{conclusiones}.}

\section{Distribution of the Quotient of Sums of Squared $\eta$-$\mu$ RVs}
\label{stat}

The fundamental results of this paper are expressed in statistical
terms in the next Lemma and its Corollary. Along this paper all
$\eta$-$\mu$ RVs are assumed to be defined by format 1, i.e.
$\eta>0$ and $\mu>0$. For details on $\eta$-$\mu$ RVs, the reader
should refer to \mbox{\cite{Yacoub07}.}

\begin{lemma}
\label{expressionA}
Let $\{X_n\}_{n=1}^N$ and $\{Y_k\}_{k=1}^K$ be \mbox{$N+K$ mutually} independent squared $\eta$-$\mu$ RVs with sets of
parameters, defined according to \mbox{format 1,} given by
$\mathcal{S}_X\equiv\{\Omega_{X_n},\eta_{X_n},\mu_{X_n}\}_{n=1}^N$ and
\mbox{$\mathcal{S}_Y\equiv\{\Omega_{Y_k},\eta_{Y_k},\mu_{Y_k}\}_{k=1}^K$} respectively.
Let us assume that $\mu_{Y_k}$ is a positive integer number for \mbox{$k=1,...,K$.}
Then, the cumulative distribution function (CDF) of the RV
$Z \triangleq {{\sum\limits_{n = 1}^N {X_n } }}/
{{\sum\limits_{k = 1}^K {Y_k } }}$
is\footnotemark[1]
\begin{equation}
\label{eq:expressionA}
\begin{gathered}
  F_Z \left( z \right) = \Theta \left( {z;{\mathcal{S}}_X ,{\mathcal{S}}_Y } \right) \triangleq  - \left( {\prod\limits_{j = 1}^J {\left( { - \beta _j } \right)^{b_j } } } \right)\left( {\prod\limits_{\ell  = 1}^{2N} {\left( {z\alpha _\ell  } \right)^{a_\ell  } } } \right)\sum\limits_{r = 1}^J {\sum\limits_{\mathbf{q} \in \wp _{2N + J}^{b_r  - 1} } {\frac{{( - 1)^{q_1 } }}
{{\beta _r ^{q_1  + 1} }} \times } }  \hfill \\
  \quad \prod\limits_{\ell  = 1}^{2N} {\frac{{( - 1)^{q_{\ell  + 1} } \left( {a_\ell  } \right)_{q_{\ell  + 1} } }}
{{\left( 1 \right)_{q_{\ell  + 1} } \left( {\beta _r  + z\alpha _\ell  } \right)^{a_\ell   + q_{\ell  + 1} } }}} \prod\limits_{j = 1}^{r - 1} {\frac{{( - 1)^{q_{_{j + 2N + 1} } } \left( {b_j } \right)_{q_{j + 2N + 1} } }}
{{\left( 1 \right)_{q_{j + 2N + 1} } \left( {\beta _r  - \beta _j } \right)^{b_j  + q_{j + 2N + 1} } }}} \prod\limits_{j = r + 1}^J {\frac{{( - 1)^{q_{_{j + 2N} } } \left( {b_j } \right)_{q_{j + 2N} } }}
{{\left( 1 \right)_{q_{j + 2N} } \left( {\beta _r  - \beta _j } \right)^{b_j  + q_{j + 2N} } }}},  \hfill \\
\end{gathered}
\end{equation}

\addtocounter{footnote}{1}\footnotetext{Along this paper it is assumed that $\prod\limits_{k = a}^b {s_k }=1$
when $b<a$.}

\noindent where $(c)_{\ell}$ is the Pochhammer symbol, the sets of coefficients
$\left\{ {a_\ell  } \right\}_{\ell  = 1}^{2N}$ and $\left\{ {\alpha _\ell  } \right\}_{\ell  = 1}^{2N}$
are defined from $\mathcal{S}_X$
as follows
\begin{equation}
\label{parame1}
\left\{ \begin{gathered}
  a_{2n - 1}  \triangleq a_{2n}  \triangleq \mu _{X_n } \hspace{23mm} \left( {n = 1, \ldots ,N} \right), \hfill \\
  \alpha _{2n - 1}  \triangleq \frac{\mu _{X_n }}
{\Omega_{X_n } }\frac{{2 + \eta _{X_n }  + \eta _{X_n }^{ - 1} }}
{{1 + \eta _{X_n } }}\hspace{8mm} \left( {n = 1, \ldots ,N} \right), \hfill \\
  \alpha _{2n}  \triangleq \eta _{X_n } \frac{\mu _{X_n }}
{\Omega_{X_n } }\frac{{2 + \eta _{X_n }  + \eta _{X_n }^{ - 1} }}
{{1 + \eta _{X_n } }}\hspace{6mm} \left( {n = 1, \ldots ,N} \right); \hfill \\
\end{gathered}  \right.
\end{equation}
the sets of coefficients
$\left\{ {b_j } \right\}_{j = 1}^J$ and $\left\{ {\beta _j } \right\}_{j = 1}^J$
are defined from $\mathcal{S}_Y$
as follows: $\left\{ {\beta _j } \right\}_{j = 1}^J$ is the set of different values in the set
of intermediate coefficients
\begin{equation*}
\mathcal{S}_Y^\star\equiv
\left\{ {\omega _k ,\rho _k } \right\}_{k = 1}^K,
\end{equation*}
where
\begin{equation}
\label{inter}
\left\{ \begin{gathered}
  \omega _k  \triangleq \frac{{\mu _{Y_k } }}
{{\Omega _{Y_k } }}\frac{{2 + \eta _{Y_k }  + \eta _{Y_k }^{ - 1} }}
{{1 + \eta _{Y_k } }}\hspace{8mm} \left( {k = 1, \ldots ,K} \right), \hfill \\
  \rho _k  \triangleq \eta _{Y_k } \frac{{\mu _{Y_k } }}
{{\Omega _{Y_k } }}\frac{{2 + \eta _{Y_k }  + \eta _{Y_k }^{ - 1} }}
{{1 + \eta _{Y_k } }}\quad \left( {k = 1, \ldots ,K} \right), \hfill \\
\end{gathered}  \right.
\end{equation}
while $J$ is the number of different values in $\mathcal{S}_Y^\star$
and
\begin{equation}
\label{losb}
b_j  \triangleq \sum\limits_{\begin{subarray}{l}
  \left\{ {k\,:\,\omega _k  = \beta _j ,\omega _k  \in {\mathcal{S}_Y^\star}} \right\} \cup  \\
  \left\{ {k\,:\,\rho _k  = \beta _j ,\rho _k  \in {\mathcal{S}_Y^\star}} \right\}
\end{subarray}}  {\mu _{Y_k } };
\end{equation}
otherwise, $\wp _{2N + J}^{b_r  - 1}$ represents the set of natural partitions of the nonnegative integer ${b_r  - 1}$
in ${2N + J}$ groups, i.e. $\wp _{2N + J}^{b_r  - 1}  \equiv \left\{ {\mathbf{q} = (q_1 ,q_2 , \cdots ,q_{2N + J} ):q_h
\in \mathbb{N}\cup \{0\},\sum\limits_{h = 1}^{2N + J} {q_h }  = b_r  - 1} \right\}$.

\end{lemma}
\begin{proof}
See Appendix I.
\end{proof}

Taking into account that the Nakagami-$m$ distribution is obtained
in an exact manner from the $\eta$-$\mu$ distribution in format 1 by setting $\mu=m$ and
$\eta\rightarrow 0$ \cite[Appendix A]{Yacoub07}, the following Corollary is derived from Lemma \ref{expressionA}.

\begin{corollary}
\label{expressionA2}
Let $\{X_n\}_{n=1}^N$ be $N$ squared $\eta$-$\mu$ RVs with set of
parameters, defined according to \mbox{format 1,} given by
$\mathcal{S}_X\equiv\{\Omega_{X_n},\eta_{X_n},\mu_{X_n}\}_{n=1}^N$.
Let $\{Y_k\}_{k=1}^K$ be $K$ squared Nakagami-$m$ RVs with set of
parameters given by $\mathcal{G}_Y\equiv\{\Omega_{Y_k},m_{Y_k}\}_{k=1}^K$.
All these $N+K$ RVs are assumed to be mutually independent.
Let us assume that $m_{Y_k}$ is a positive integer number for $k=1,...,K$.
Then, the CDF of $Z \triangleq{{\sum\limits_{n = 1}^N {X_n } }}/{{\sum\limits_{k = 1}^K {Y_k } }}$
is given by
\begin{equation}
\label{eq:expressionA2}
F_Z \left( z \right) = \mathop {\lim }\limits_{\eta _{Y_1 } ,...,\eta _{Y_k  \to 0} }
\Theta \left( {z;{\mathcal{S}}_X ,\left\{ {\Omega _{Y_k } ,m_{Y_k } ,\eta _{Y_k } } \right\}_{k = 1}^K } \right)
= \widetilde{\Theta} \left( {z;{\mathcal{S}}_X ,{\mathcal{G}}_Y } \right),
\end{equation}
where $\widetilde\Theta$ is a slightly modified version of the function
$\Theta$ which has the same formal expression given in (\ref{eq:expressionA}) but with the only difference that
the sets of coefficients $\left\{ {b_j } \right\}_{j = 1}^J$ and $\left\{ {\beta _j } \right\}_{j = 1}^J$ are
defined from $\mathcal{G}_Y^\star\equiv\left\{ {\rho _k=\frac{{m_{Y_k} }}{{\Omega _{Y_k } }} } \right\}_{k = 1}^K$ (instead of
$\mathcal{S}_Y^\star$).
\end{corollary}
\begin{proof}
See Appendix II.
\end{proof}

It will be shown in next Section that formal expression (1), after an appropriate redefinition
of its coefficients, can be also applied to certain cases in which
the $\eta$-$\mu$ RVs in the numerator of the quotient in Lemma \ref{expressionA} are statistically correlated.

\section{Outage Probability Analysis}
\label{analysis}

In \mbox{Subsection \ref{analysis1}}
the \mbox{$\eta$-$\mu$/$\eta$-$\mu$} interference-limited scenario with independent MRC branches is analyzed;
several important particular cases are addressed, e.g. the $\eta$-$\mu$/Rayleigh and the Nakagami-$m$/Nakagami-$m$ scenarios.
Finally, the extension of the OP analysis to the spatially correlated MRC case is tackled
in Subsection \ref{analysis2}. A summary of the most relevant results derived in this Section is shown in Table I.

\subsection{Interference-limited scenarios with independent MRC }
\label{analysis1}

Let us assume a general interference-limited scenario where both the signal of
interest (SOI) and the interfering signals experience fading.
The SOI is formed from $N$ independent MRC branches with power envelopes $X_n$ (${n=1,...,N}$)
while
the interferers are received with power envelopes
$Y_k$ (${k=1,...,K}$).
All the received power envelopes are assumed to be mutually independent.
The OP for this general scenario
is defined as
\begin{equation}
\label{eq:OP}
P_{out}  \triangleq \Pr \left[ {\frac{{\sum\limits_{n = 1}^N {X_n } }}
{{\sum\limits_{k = 1}^K {Y_k } }} \leqslant \zeta _o } \right],
\end{equation}
where $\zeta _o$ is a predefined threshold. Lemma \ref{expressionA} and its Corollary are now exploited to
obtain elementary exact closed-form expressions for the OP defined in (\ref{eq:OP}). The connection of these
novel results with the well-known expressions \cite[eq. 10.17]{Simon05} and \cite[eq. 18]{Romero07}
for the \mbox{Nakagami-$m$/Nakagami-$m$} scenario is also discussed.

It is important to notice that the Nakagami-$m$ distribution is obtained
in an exact manner from the $\eta$-$\mu$ distribution in format 1 by setting $\mu=m$ and
$\eta\rightarrow 0$, or alternatively, by setting $\mu=m/2$ and setting
$\eta\rightarrow 1$ \cite[Appendix A]{Yacoub07}.
Since our analysis assumes positive integer values for the $\mu$ parameter of the CCI signals, we must distinguish between the
following two alternative types of expressions for the OP.

\subsubsection{\mbox{Type I, $\eta$-$\mu$/$\eta$-$\mu$} interference-limited scenario}

Let $\mathcal{S}_X\equiv\{\Omega_{X_n},\eta_{X_n},\mu_{X_n}\}_{n=1}^N$ be the set of parameters
for the pre-combining SOI power envelopes, and $\mathcal{S}_Y\equiv\{\Omega_{Y_k},\eta_{Y_k},\mu_{Y_k}\}_{k=1}^K$
the corresponding set of parameters for the CCI power envelopes.
Both $\eta$-$\mu$ parameters are defined according to \mbox{format 1.}
The only restriction in the values of the
problem parameters is that $\mu_{Y_k}$ (${k=1,...,K}$) are assumed to be positive integer numbers.
Then, for this scenario the OP is given by
\begin{equation}
\label{eq:OP1}
P_{out}  = \Theta \left( {\zeta _o ;\mathcal{S}_X ,\mathcal{S}_Y } \right),
\end{equation}
where the elementary function $\Theta$ is defined in Lemma \ref{expressionA}.
Note that expression (\ref{eq:OP1}) includes the \mbox{$\eta$-$\mu$/Nakagami-$m$} scenario,
when the $m$ parameters of the
CCI signals are even positive integers numbers, by setting $\eta_{Y_k}=1$ and $\mu_{Y_k}=m_{Y_k}/2$ for
$k=1,\ldots,K$.
Otherwise, the OP for the \mbox{Nakagami-$m$/$\eta$-$\mu$} scenario is obtained,
for arbitrary values of the corresponding $m$ parameters,
by setting $\eta_{X_n}=1$ and $\mu_{X_n}=m_{X_n}/2$ for $n=1,\ldots,N$; thus,
the OP for the \mbox{Rayleigh/$\eta$-$\mu$},
\mbox{Hoyt/$\eta$-$\mu$} and \mbox{one-sided Gaussian/$\eta$-$\mu$} scenarios are also included in expression
(\ref{eq:OP1}) for arbitrary values of the corresponding statistical parameters of the pre-combining SOI power envelopes.

\subsubsection{\mbox{Type II, $\eta$-$\mu$/Nakagami-$m$} interference-limited scenario}

Let $\mathcal{S}_X\equiv\{\Omega_{X_n},\eta_{X_n},\mu_{X_n}\}_{n=1}^N$ be the set of parameters
for the pre-combining SOI power envelopes $\{X_n\}_{n=1}^N$, and
\mbox{$\mathcal{G}_Y\equiv\{\Omega_{Y_k},m_{Y_k}\}_{k=1}^K$}
the corresponding set of parameters for the CCI power envelopes.
The $\eta$-$\mu$ parameters for the pre-combining SOI power envelopes are defined according to \mbox{format 1.}
The only restriction in the values of the
problem parameters is that $m_{Y_k}$ (${k=1,...,K}$) are assumed to be positive integer numbers.
Then, for this scenario the OP is given by
\begin{equation}
\label{eq:OP2}
P_{out}  = \widetilde{\Theta} \left( {\zeta _o ;\mathcal{S}_X ,\mathcal{G}_Y } \right),
\end{equation}
where the elementary function $\widetilde{\Theta}$ is defined in Corollary \ref{expressionA2}.
Expression (\ref{eq:OP2}) includes the OP for
the \mbox{Nakagami-$m$/Nakagami-$m$} scenario
with arbitrary values for the $m$ parameters corresponding to the pre-combining SOI power envelopes.
This is possible by setting $\eta_{X_n}=1$ and $\mu_{X_n}=m_{X_n}/2$ for $n=1,\ldots,N$; thus,
expression (\ref{eq:OP2}) complements both \cite[eq. 10.17]{Simon05} and \cite[eq. 18]{Romero07}
as long as it is valid if real values of $m$ are assumed for the pre-combining SOI power envelopes,
with the only restriction that the values of $m$ for the interferers are positive integer numbers.
Note that the OP for the $\eta$-$\mu$/Rayleigh interference-limited scenario without restrictions on the
corresponding statistical parameters is also included in (\ref{eq:OP2}), by setting $m_{Y_k}=1$ for
$k=1,\ldots,K$.

\subsection{Extension to spatially correlated MRC}
\label{analysis2}

It is shown in this Section that all scenarios analyzed by the OP expressions (\ref{eq:OP1}) and (\ref{eq:OP2})
can be generalized to scenarios in which spatially correlated MRC is considered for the SOI.
Expression (\ref{eq:OP1}) and (\ref{eq:OP2}) where derived with the only restriction that the $\mu$ parameters
for the CCI signals must be positive integer numbers. To carry out our extension, an additional restriction is required for
the $\mu$ parameters of the pre-combining SOI signals. Specifically, we assume in this Section that
$\mu_{X_n}$ (${n=1,...,N}$) are assumed to be positive integer or half-integer numbers; however, this restriction still allows us
to include in our analysis Nakagami-$m$ fading with integer or \mbox{half-integer $m$} for the pre-combining SOI signals,
 and in particular, Rayleigh, Hoyt and one-sided Gaussian fading.

Very recently, an elegant expression for the moment generating function (MGF) of the received power
$\sum\limits_{n = 1}^N {X_n }$ in correlated MRC under $\eta$-$\mu$ fading was derived in \cite[eq. 12]{Ashgari10}
\begin{equation}
\label{eq:corr1}
{\mathcal{M}}_{\sum\limits_{n = 1}^N {X_n } } \left( s \right) = \prod\limits_{v = 1}^V {\left( {1 - 2\lambda _v^X s} \right)^{ - \frac{{\xi _v^Y }}
{2}} \left( {1 - 2\lambda _v^Y s} \right)^{ - \frac{{\xi _v^X }}
{2}} },
\end{equation}
where $\left\{ {\lambda _v^X } \right\}_{v = 1}^V$ and $\left\{ {\lambda _v^Y } \right\}_{v = 1}^V$
represent distinct eigenvalues, with $\left\{ {\xi _v^X } \right\}_{v = 1}^V$ and
$\left\{ {\xi _v^Y } \right\}_{v = 1}^V$ their corresponding algebraic multiplicities,
of certain covariance matrices defined in \cite[sect. III]{Ashgari10}
which contain the spatial correlation structure of the channel.
Expression (\ref{eq:corr1}) requires that $\mu_{X_n}$ (${n=1,...,N}$) are assumed to be positive
integer or half-integer numbers.

It is shown along the proof of Lemma \ref{expressionA} given in Appendix I that the MGF of the received power
$\sum\limits_{n = 1}^N {X_n }$ under independent MRC can be formally expressed as
\begin{equation}
\label{eq:corr2}
{\mathcal{M}}_{\sum\limits_{n = 1}^N {X_n } } \left( s \right) = \prod\limits_{v = 1}^N {\left( {1 - \frac{s}
{{\alpha _{2v - 1} }}} \right)^{ - a_{2v - 1} } \left( {1 - \frac{s}
{{\alpha _{2v} }}} \right)^{ - a_{2v} } };
\end{equation}
thus, a simple comparison of (\ref{eq:corr1}) and (\ref{eq:corr2}) allows us to infer that the formal replacement
\begin{equation}
\label{eq:corr3}
\begin{gathered}
  V \leftrightarrow N \hfill \\
  \frac{1}
{{2\lambda _v^X }} \leftrightarrow \alpha _{2v - 1}, \hspace{10mm}
  \frac{1}
{{2\lambda _v^Y }} \leftrightarrow \alpha _{2v} \hspace{8mm}\left( {v = 1, \ldots ,V} \right),\hfill \\
  \frac{{\xi _v^X }}
{2} \leftrightarrow a_{2v - 1}, \hspace{14mm}
  \frac{{\xi _v^Y }}
{2} \leftrightarrow a_{2v} \hspace{11mm} \left( {v = 1, \ldots ,V} \right), \hfill \\
\end{gathered}
\end{equation}
yields the following closed-form expression for the OP under spatially correlated MRC
\begin{equation}
\label{eq:corr4}
P_{out}  = \Theta \left( {\zeta_o;\left\{ {a_{2v - 1}  = \frac{{\xi _v^X }}
{2},a_v  = \frac{{\xi _v^Y }}
{2}} \right\}_{v = 1}^V ,\left\{ {\alpha _{2v - 1}  = \frac{1}
{{2\lambda _v^X }},\alpha _v  = \frac{1}
{{2\lambda _v^X }}} \right\}_{v = 1}^V ;{\mathcal{S}}_Y } \right).
\end{equation}
With the following notation
$\Theta \left( {z;\left\{ {a_n } \right\}_{n = 1}^{2N} ,\left\{ {\alpha _n } \right\}_{n = 1}^{2N} ;\mathcal{S}_Y } \right)$
we represent a function that
has the same formal expression given in ($\ref{eq:expressionA}$), with the
set of coefficients $\left\{ {a_n } \right\}_{n = 1}^{2N} $ and $\left\{ {\alpha _n } \right\}_{n = 1}^{2N}$
directly given as inputs parameters (they are not calculated by (\ref{parame1})) and the remainder
coefficients which appear in (\ref{eq:expressionA}) are calculated as in Lemma \ref{expressionA}.
Note that expression (\ref{eq:corr4}) provides an extension of the type I expression given in (\ref{eq:OP1})
for spatially correlated MRC, with the additional assumption that $\mu_{X_n}$ (${n=1,...,N}$)
are assumed to be positive integer or half-integer numbers. In connection with the type II expression (\ref{eq:OP2}) we can
also obtain
\begin{equation}
\label{eq:corr5}
P_{out}  = \widetilde{\Theta} \left( {\zeta_o;\left\{ {a_{2v - 1}  = \frac{{\xi _v^X }}
{2},a_v  = \frac{{\xi _v^Y }}
{2}} \right\}_{v = 1}^V ,\left\{ {\alpha _{2v - 1}  = \frac{1}
{{2\lambda _v^X }},\alpha _v  = \frac{1}
{{2\lambda _v^X }}} \right\}_{v = 1}^V ;\mathcal{G}_Y } \right),
\end{equation}
where $\widetilde{\Theta} \left( {z;\left\{ {a_n } \right\}_{n = 1}^{2N} ,\left\{ {\alpha _n } \right\}_{n = 1}^{2N} ;\mathcal{G}_Y } \right)$
represents a function that
has the same formal expression given in ($\ref{eq:expressionA}$), with the
set of coefficients $\left\{ {a_n } \right\}_{n = 1}^{2N} $ and $\left\{ {\alpha _n } \right\}_{n = 1}^{2N}$
directly given as inputs parameters (they are not calculated by (\ref{parame1})) and the remainder
coefficients which appear in (\ref{eq:expressionA}) are calculated as in Corollary \ref{expressionA2}.
Again, the additional assumption that $\mu_{X_n}$ (${n=1,...,N}$) are positive
integer or half-integer numbers is required.

\section{Numerical Results}
\label{resultados}

Fig. 1 shows some numerical results obtained by the type I OP expression
given in (\ref{eq:OP1}) for \mbox{$\zeta_o=10$.}
In this first example, the SOI is formed from three $\eta$-$\mu$ power envelopes ($N=3$) with set of parameters
$\mathcal{S}_X\equiv\{\Omega_{X_n},\eta_{X_n},\mu_{X_n}\}_{n=1}^3=\{\{\Omega,2.6,\mu\},\{0.8\Omega,3.4,\mu\},\{0.7\Omega,1.7,\mu\}\}$;
and three $\eta$-$\mu$ interferers ($K=3$) with
$\mathcal{S}_Y\equiv\{\Omega_{Y_k},\eta_{Y_k},\mu_{Y_k}\}_{k=1}^3=\{\{1,3.3,2\},\{1,3.3,2\},\{0.5,1.7,1\}\}$
are considered.
Different plots are shown in Fig. 1 in terms of the average signal-to-interference ratio (SIR),
which is $\Omega$ in this example, for different values of the parameter $\mu$. Fig. 2 plots the OP
for some scenarios included in the type II expression (\ref{eq:OP2}) with $\zeta_o=10$.
In this second example,
the SOI is formed from two $\eta$-$\mu$ power envelopes ($N=2$) with set of parameters
$\mathcal{S}_X\equiv\{\Omega_{X_n},\eta_{X_n},\mu_{X_n}\}_{n=1}^2
=\{\{2\Omega,1,0.5\},\{0.7\Omega,0.6,2\}\}$;
and four Nakagami-$m$ interferers \mbox{($K=4$)} with
$\mathcal{G}_Y\equiv\{\Omega_{Y_k},m_{Y_k}\}_{k=1}^4=\{\{1,m\},\{1,m\},\{0.5,m\},\{0.2,m\}\}$
are considered. Again, OP curves are shown in Fig. 1 in terms of the average SIR,
which is equal to $\Omega$, for different values of the parameter $m$.

Simulation results are also superimposed in both figures.

\section{Conclusions}
\label{conclusiones}

The master formula (\ref{eq:expressionA}) is derived in this
paper, which by appropriate setting of its formal parameters
allows us to obtain a variety of elementary closed-form
expressions for the outage probability in
$\eta$-$\mu$/$\eta$-$\mu$ interference limited scenarios.

%

\appendices

\section{Proof of Lemma I}

Let $\L[f(x);x,s]$ or simply $\L[f(x);s]$ represent the Laplace transform of $f(x)$ defined as
\begin{equation}
\label{app1}
\L[f\left( x \right);s] \triangleq \int_0^\infty  {e^{ - sx} f\left( x \right)dx}.
\end{equation}
If $f(x)$ is the probability density function (PDF) of a RV $X$, then its MGF $\M_X (s)$ is
defined as $\M_X (s)\triangleq \L[f\left( x \right);-s]$.

According to \cite[eq. 6]{Ermolova10}, the MGF for the sum of independent but arbitrarily
distributed $\eta$-$\mu$ RVs $\{X_n\}_{n=1}^N$, defined in format 1, can be expressed as
\begin{equation}
\label{app2}
\M_{\sum\limits_{n = 1}^N {X_n } } \left( s \right) = \prod\limits_{n = 1}^N {\left( {1 - \frac{s}
{{\alpha _{2n - 1} }}} \right)^{ - a_{2n - 1} } \left( {1 - \frac{s}
{{\alpha _{2n} }}} \right)^{ - a_{2n} } },
\end{equation}
where the sets of coefficients
$\left\{ {a_\ell  } \right\}_{\ell  = 1}^{2N}$ and $\left\{ {\alpha _\ell  } \right\}_{\ell  = 1}^{2N}$
are defined from $\mathcal{S}_X$
as specified in (\ref{parame1}). Both sets of coefficients are formed by positive real numbers which
are not necessarily distinct. The same idea allows us to express the MGF of the sum of the
$\eta$-$\mu$ RVs $\{Y_k\}_{k=1}^K$ as
\begin{equation}
\label{app3}
\M_{\sum\limits_{k = 1}^K {Y_k } } \left( s \right) = \prod\limits_{k = 1}^K {\left( {1 - \frac{s}
{{\omega _k }}} \right)^{ - \mu _{Y_k } } \left( {1 - \frac{s}
{{\rho _k }}} \right)^{ - \mu _{Y_k } } },
\end{equation}
where the set of intermediate coefficients
$\mathcal{S}_Y^\star\equiv
\left\{ {\omega _k ,\rho _k } \right\}_{k = 1}^K$
is defined according to (\ref{inter}).
By hypothesis, the set of parameters $\{{\mu}_{Y_k}\}_{k=0}^K$ are positive integers; thus, expression
(\ref{app3}) can be rearranged as
\begin{equation}
\label{app4}
\M_{\sum\limits_{k = 1}^K {Y_k } } \left( s \right) = \prod\limits_{j = 1}^J {\left( {1 - \frac{s}
{{\beta _j }}} \right)^{ - b_j } },
\end{equation}
where $\left\{ {\beta _j } \right\}_{j = 1}^J$ is the set of different values in the set
of intermediate coefficients $\mathcal{S}_Y^\star$, $J$ is the number of different values in $\mathcal{S}_Y^\star$
and $\left\{ b _j  \right\}_{j = 1}^J$ is the set of positive integers defined by
(\ref{losb}) representing the multiplicities corresponding to $\left\{ {\beta _j } \right\}_{j = 1}^J$.

The CDF of $Z ={{\sum\limits_{n = 1}^N {X_n } }}/{{\sum\limits_{k = 1}^K {Y_k } }}$ can be expressed through the
MGFs of $\sum\limits_{n = 1}^N {Y_n }$ and $\sum\limits_{k = 1}^K {Y_k }$ by appealing to the well-known
convolution and scaling properties of the Laplace transform, i.e.
\begin{equation}
\label{app5}
\begin{gathered}
  F_Z \left( z \right) = \Pr \left[ {\sum\limits_{n = 1}^N {X_n }  \leqslant z\sum\limits_{k = 1}^K {Y_k } } \right] =  \hfill \\
  \quad \quad \int_0^\infty  {F_{\sum\limits_{n = 1}^N {X_n } } \left( {zy} \right)f_{\sum\limits_{k = 1}^K {Y_k } } \left( y \right)dy}  = \left. {\mathcal{L}\left[ {F_{\sum\limits_{n = 1}^N {X_n } } \left( {zy} \right)f_{\sum\limits_{k = 1}^K {Y_k } } \left( y \right);y,s} \right]} \right|_{s = 0}  =  \hfill \\
  \quad \quad \frac{1}
{{2\pi i}}\int\limits_{c - i\infty }^{c + i\infty } {\mathcal{L}\left[ {F_{\sum\limits_{n = 1}^N {X_n } } \left( {zy} \right);p} \right]}
 \mathcal{L}\left[ {f_{\sum\limits_{k = 1}^K {Y_k } } \left( x \right); - p} \right]dp =  \hfill \\
  \quad \quad \frac{1}
{{2\pi i}}\int\limits_{c - i\infty }^{c + i\infty } {\frac{1}
{z}\mathcal{L}\left[ {F_{\sum\limits_{n = 1}^N {X_n } } \left( y \right);y,\frac{p}
{z}} \right]} \mathcal{L}\left[ {f_{\sum\limits_{k = 1}^K {Y_k } } \left( y \right);y, - p} \right]dz =  \hfill \\
  \quad \quad \frac{1}
{{2\pi i}}\int\limits_{c - i\infty }^{c + i\infty } {\Xi \left( p \right)} dp, \hfill \\
\end{gathered}
\end{equation}
where the integration kernel
\begin{equation}
\label{app6}
\Xi \left( p \right) \triangleq \frac{1}
{p}\mathcal{M}_{\sum\limits_{n = 1}^N {X_n } } \left( { - \frac{p}
{z}} \right)\mathcal{M}_{\sum\limits_{k = 1}^K {Y_k } } \left( p \right),
\end{equation}
$i$ is the imaginary unit and $c$ an appropriate real number which splits the singularities of the involved MGFs.
Considering (\ref{app2}) and (\ref{app4}), the integration kernel in our case is expressed as
\begin{equation}
\label{app7}
\Xi \left( p \right) = \frac{1}
{p}\prod\limits_{\ell  = 1}^{2N} {\left( {1 + \frac{p}
{{z\alpha _\ell  }}} \right)} ^{ - a_\ell  } \prod\limits_{j = 1}^J {\left( {1 - \frac{p}
{{\beta _j }}} \right)} ^{ - b_j }.
\end{equation}
Let $\left\{ {\hat{\alpha}_\ell  } \right\}_{\ell  = 1}^{2N}$ and $\left\{ {\hat{\beta}_j } \right\}_{j = 1}^J$
denote the complex-modulus ordered sets corresponding to $\left\{ {\alpha_\ell  } \right\}_{\ell  = 1}^{2N}$
and $\left\{ {{\beta}_j } \right\}_{j = 1}^J$ respectively, i.e.
$0 \leqslant \left| {\hat \alpha _1 } \right| \leqslant \left| {\hat \alpha _2 } \right| \leqslant  \ldots  \leqslant \left| {\hat \alpha _{2N} } \right|$
and
$0 < \left| {\hat \beta _1 } \right| < \left| {\hat \beta _2 } \right| <  \ldots  < \left| {\hat \beta _J } \right|$;
then, the singularity structure of the integration kernel $\Xi(p)$ and the integration paths involved
in the computation of the CDF of $Z$ are shown in Fig. \ref{fig3}.
Since $
\left| {\left. {\Xi \left( p \right)} \right|_{\wp _\infty  } } \right| \leqslant \frac{C}
{R}\prod\limits_{\ell  = 1}^{2N} {\left( {\frac{R}
{{z\alpha _\ell  }}} \right)} ^{ - a_\ell  } \prod\limits_{j = 1}^J {\left( {\frac{R}
{{\beta _j }}} \right)} ^{ - b_j }
$ for a sufficiently large $R$ and an appropriate constant $C$, then
\begin{equation}
\label{app8}
\mathop {\lim }\limits_{R \to \infty } \left| {\int_{\wp _\infty  } {\Xi \left( p \right)dp} } \right| \leqslant
\mathop {\lim }\limits_{R \to \infty } 2\pi R\left| {\left. {\Xi \left( p \right)} \right|_{\wp _\infty  } } \right| = 0.
\end{equation}
Thus, the Cauchy-Goursat and the Residue Theorems allows us to express the CDF of $Z$ as follows
\begin{equation}
\label{app9}
\begin{gathered}
  F_Z (z) = \frac{1}
{{2\pi i}}\int\limits_{\varepsilon  - i\infty }^{\varepsilon  + i\infty } {\Xi \left( p \right)} dp = \frac{1}
{{2\pi i}}\sum\limits_{j = 1}^J {\int_{\mathcal{C}\left( {\beta _j } \right)} {\Xi \left( p \right)dp} }  - \frac{1}
{{2\pi i}}\mathop {\lim }\limits_{R \to \infty } \int_{\wp _\infty  } {\Xi \left( p \right)dp}  =  - \sum\limits_{r = 1}^J {\text{Res}\left[ {\Xi ;\beta _r } \right]}  \hfill \\
  \quad  =  - \left( {\prod\limits_{j = 1}^J {\left( { - \beta _j } \right)^{b_j } } } \right)\left( {\prod\limits_{\ell  = 1}^{2N} {\left( {z\alpha _\ell  } \right)^{a_\ell  } } } \right)\sum\limits_{r = 1}^J {\frac{1}
{{(b_r  - 1)!}}\left. {\frac{{d^{b_r  - 1} }}
{{dp^{b_r  - 1} }}} \right|_{\beta _r } \left( {\frac{1}
{p}\prod\limits_{\ell  = 1}^{2N} {\left( {p + z\alpha _\ell  } \right)^{ - a_\ell  } } \prod\limits_{\begin{subarray}{l}
  j = 1 \\
  j \ne r
\end{subarray}} ^J {\left( {p_r  - \beta _j } \right)^{ - b_j } } } \right)},  \hfill \\
\end{gathered}
\end{equation}
where $0 < \varepsilon  < \left| {\hat \beta _1 } \right|$ and
${\text{Res}\left[ {\Xi;\beta _r } \right]}$ represents the residue of $\Xi(p)$ at
the pole $p=\beta _r$. Finally, taking into account
the Leibniz derivative rule for products and that
\begin{equation}
\label{app10}
\frac{1}
{{q!}}\frac{{d^q }}
{{dx^q }}\left( {x + a} \right)^{ - m}  = \frac{{( - 1)^q \left( m \right)_q }}
{{\left( 1 \right)_q \left( {x + a} \right)^{m + q} }},
\end{equation}
the desired expression (\ref{eq:expressionA})
is obtained after simple algebraic manipulations.


\section{Proof of Corollary I}

It is clear from (\ref{inter}) that ${\omega _k\rightarrow \infty }$
and ${\rho _k\rightarrow\frac{{m_{Y_k} }}{{\Omega _{Y_k } }} }$ when $\eta_{Y_k}\rightarrow 0 $
and $\mu_{Y_k}=m_{Y_k} $ for $k=1,\ldots,K$; thus, in this case
\begin{equation}
\label{app11}
\mathcal{M}_{\sum\limits_{j = 1}^K {Y_k } } \left( s \right) = \prod\limits_{k = 1}^K {\left( {1 - \frac{s}
{{\rho _k }}} \right)^{ - m_{Y_k } } },
\end{equation}
and the same formal proof given in Lemma I can be repeated
from expression (\ref{app4}) by replacing ${\mathcal{S}_Y^{\star}}$ with ${\mathcal{G}_Y^{\star}}$.

\newpage

\begin{table}[t!]
\caption{Summary of Some Results Derived in Section \ref{analysis}.}
\begin{center}
\small{\begin{tabular}{c | c | c | c | c | c | c | c | c }

\hline\hline
\multicolumn{3}{ c }{Scenario} & \multicolumn{3}{ c }{Assumptions on the Statistical Parameters}
& \multicolumn{3}{ c }{OP Expression}\\

\hline
\hline
\multicolumn{3}{c |}{
\begin{minipage}{3 cm}
\vspace{2mm}
\begin{center}
$\eta$-$\mu$/$\eta$-$\mu$\\
Independent MRC
\end{center}
\vspace{0mm}
\end{minipage}
}
& \multicolumn{3}{| c |}{
\begin{minipage}{10 cm}
\vspace{2mm}
\begin{itemize}
\item The $\mu$ parameters of the CCI signals are positive integers.
\vspace{2mm}
\end{itemize}
\end{minipage}
}
& \multicolumn{3}{| c }{(\ref{eq:OP1})}\\
\hline
\multicolumn{3}{c |}{
\begin{minipage}{3 cm}
\vspace{2mm}
\begin{center}
Nakagami-$m$/$\eta$-$\mu$\\
Independent MRC
\end{center}
\vspace{0mm}
\end{minipage}
}
& \multicolumn{3}{| c |}{
\begin{minipage}{10 cm}
\vspace{2mm}
\begin{itemize}
\item The $\mu$ parameters of the CCI signals are positive integers.
\vspace{2mm}
\end{itemize}
\end{minipage}
}
& \multicolumn{3}{| c }{(\ref{eq:OP1})}\\
\hline
\multicolumn{3}{c |}{
\begin{minipage}{3 cm}
\vspace{2mm}
\begin{center}
$\eta$-$\mu$/Nakagami-$m$\\
Independent MRC
\end{center}
\vspace{0mm}
\end{minipage}
}
& \multicolumn{3}{| c |}{
\begin{minipage}{10 cm}
\vspace{2mm}
\begin{itemize}
\item The $m$ parameters of the CCI signals are positive integers.
\vspace{2mm}
\end{itemize}
\end{minipage}
}
& \multicolumn{3}{| c }{(\ref{eq:OP2})}\\
\hline
\multicolumn{3}{c |}{
\begin{minipage}{3 cm}
\vspace{2mm}
\begin{center}
Nakagami-$m$/\\/Nakagami-$m$\\
Independent MRC
\end{center}
\vspace{0mm}
\end{minipage}
}
& \multicolumn{3}{| c |}{
\begin{minipage}{10 cm}
\vspace{2mm}
\begin{itemize}
\item The $m$ parameters of the CCI signals are positive integers.
\vspace{2mm}
\end{itemize}
\end{minipage}
}
& \multicolumn{3}{| c }{(\ref{eq:OP2})}\\
\hline
\multicolumn{3}{c |}{
\begin{minipage}{3 cm}
\vspace{2mm}
\begin{center}
$\eta$-$\mu$/Rayleigh\\
Independent MRC
\end{center}
\vspace{0mm}
\end{minipage}
}
& \multicolumn{3}{| c |}{
\begin{minipage}{10 cm}
\vspace{2mm}
\begin{itemize}
\item None.
\vspace{2mm}
\end{itemize}
\end{minipage}
}
& \multicolumn{3}{| c }{(\ref{eq:OP2})}\\
\hline
\multicolumn{3}{c |}{
\begin{minipage}{3 cm}
\vspace{2mm}
\begin{center}
$\eta$-$\mu$/$\eta$-$\mu$\\
Correlated MRC
\end{center}
\vspace{0mm}
\end{minipage}
}
& \multicolumn{3}{| c |}{
\begin{minipage}{10 cm}
\vspace{2mm}
\begin{itemize}
\item The $\mu$ parameters of the pre-combining SOI signals are positive integers or half-integers.
\item The $\mu$ parameters of the CCI signals are positive integers.
\vspace{2mm}
\end{itemize}
\end{minipage}
}
& \multicolumn{3}{| c }{(\ref{eq:corr4})}\\
\hline
\multicolumn{3}{c |}{
\begin{minipage}{3 cm}
\vspace{2mm}
\begin{center}
Nakagami-$m$/\\/Nakagami-$m$\\
Correlated MRC
\end{center}
\vspace{0mm}
\end{minipage}
}
& \multicolumn{3}{| c |}{
\begin{minipage}{10 cm}
\vspace{2mm}
\begin{itemize}
\item The $m$ parameters of the pre-combining SOI signals are positive integers or half-integers.
\item The $m$ parameters of the CCI signals are positive integers.
\vspace{2mm}
\end{itemize}
\end{minipage}
}
& \multicolumn{3}{| c }{(\ref{eq:corr5})}\\
\hline
\multicolumn{3}{c |}{
\begin{minipage}{3 cm}
\vspace{2mm}
\begin{center}
$\eta$-$\mu$/Rayleigh\\
Correlated MRC
\end{center}
\vspace{0mm}
\end{minipage}
}
& \multicolumn{3}{| c |}{
\begin{minipage}{10 cm}
\vspace{2mm}
\begin{itemize}
\item The $\mu$ parameters of the pre-combining SOI signals are positive integers or half-integers.
\vspace{2mm}
\end{itemize}
\end{minipage}
}
& \multicolumn{3}{| c }{(\ref{eq:corr5})}\\

\hline\hline

\end{tabular}}
\small{\begin{tabular}{l}
\end{tabular}}
\end{center}
\end{table}

\newpage

\begin{figure}
\centering \vspace*{45mm}
\includegraphics[width=14cm]{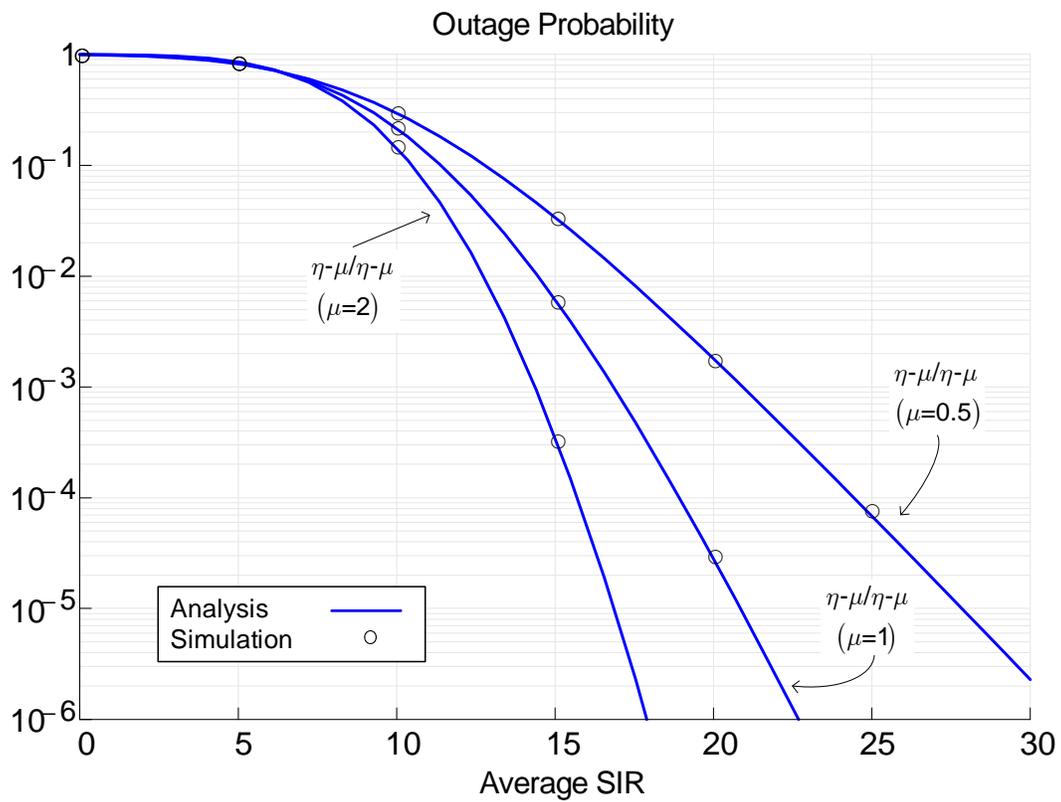}
\begin{center}
\begin{minipage}{\figtab}
\caption{ \footnotesize Type I OP expression (\ref{eq:OP1}) versus average SIR $\Omega$ for different values of $\mu$,
where
$\mathcal{S}_X\equiv\{\Omega_{X_n},\eta_{X_n},\mu_{X_n}\}_{n=1}^3=\{\{\Omega,2.6,\mu\},\{0.8\Omega,3.4,\mu\},\{0.7\Omega,1.7,\mu\}\}$,
$\mathcal{S}_Y\equiv\{\Omega_{Y_k},\eta_{Y_k},\mu_{Y_k}\}_{k=1}^3=\{\{1,3.3,2\},\{1,3.3,2\},\{0.5,1.7,1\}\}$
and $\zeta_o=10$.
}
  \label{fig1}
\end{minipage} \figspace
\end{center}
\end{figure}

\newpage

\begin{figure}
\centering \vspace*{45mm}
\includegraphics[width=14cm]{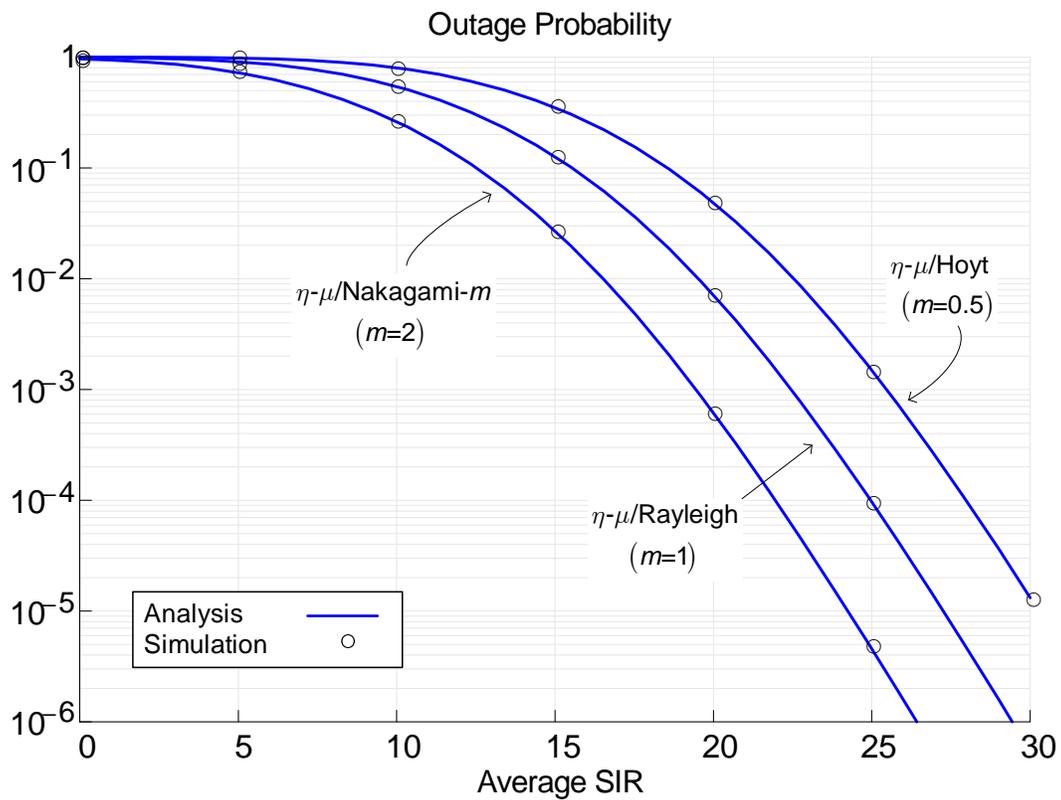}
\begin{center}
\begin{minipage}{\figtab}
\caption{ \footnotesize Type II OP expression (\ref{eq:OP2}) versus average SIR $\Omega$ for different values of $m$,
where
$\mathcal{S}_X\equiv\{\Omega_{X_n},\eta_{X_n},\mu_{X_n}\}_{n=1}^2
=\{\{2\Omega,1,0.5\},\{0.7\Omega,0.6,2\}\}$,
$\mathcal{G}_Y\equiv\{\Omega_{Y_k},m_{Y_k}\}_{k=1}^4=\{\{1,m\},\{1,m\},\{0.5,m\},\{0.2,m\}\}$
and $\zeta_o=10$. }
  \label{fig2}
\end{minipage} \figspace
\end{center}
\end{figure}

\newpage

\begin{figure}
\centering \vspace*{45mm}
\includegraphics[width=14cm]{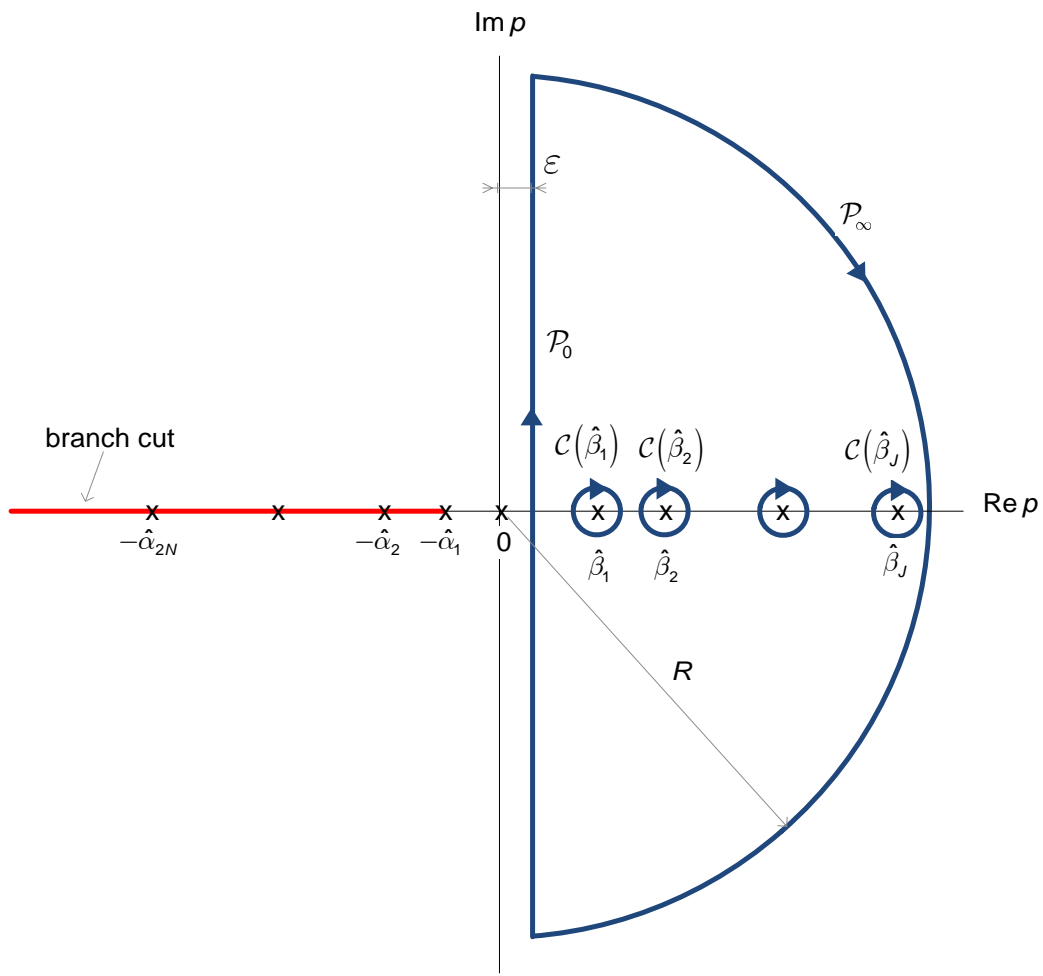}
\begin{center}
\begin{minipage}{\figtab}
\caption{ \footnotesize Singularity structure of the integration kernel $\Xi(p)$ and integration paths involved in the proof of Lemma 1. }
  \label{fig3}
\end{minipage} \figspace
\end{center}
\end{figure}

\end{document}